\documentclass[9pt]{article}
\usepackage{spconf,amsmath,graphicx,hyperref}
\usepackage{times}
\usepackage{epsfig}
\usepackage{amsthm}
\usepackage{amssymb}
\usepackage{algorithm}
\usepackage{algorithmic}
\usepackage{enumerate}
\usepackage{multirow}
\usepackage{multicol}
\usepackage{arydshln}
\usepackage{color}
\usepackage{amssymb}
\usepackage{soul}
\usepackage{subfig}
\usepackage{subcaption}
\usepackage{caption}
\usepackage{float}
\usepackage{booktabs}
\usepackage{babel,blindtext}

\newtheorem{theorem}{Theorem}
\newtheorem{lemma}{Lemma}

\def\0{{\mathbf 0}}
\def\1{{\mathbf 1}}

\def\f{{\mathbf f}}

\def\r{{\mathbf r}}

\def\u{{\mathbf u}}
\def\v{{\mathbf v}}
\def\w{{\mathbf w}}
\def\x{{\mathbf x}}
\def\y{{\mathbf y}}
\def\z{{\mathbf z}}

\def\D{{\mathbf D}}

\def\H{{\mathbf H}}
\def\I{{\mathbf I}}

\def\L{{\mathbf L}}
\def\M{{\mathbf M}}

\def\R{{\mathbf R}}
\def\S{{\mathbf S}}
\def\T{{\mathbf T}}
\def\U{{\mathbf U}}
\def\V{{\mathbf V}}
\def\W{{\mathbf W}}

\def\ie{{\textit{i.e.}}}

\def\cE{{\mathcal E}}

\def\cG{{\mathcal G}}

\def\cK{{\mathcal K}}
\def\cL{{\mathcal L}}

\def\cN{{\mathcal N}}
\def\cO{{\mathcal O}}
\def\cP{{\mathcal P}}

\def\cS{{\mathcal S}}

\def\cV{{\mathcal V}}

\def\balpha{{\boldsymbol \alpha}}

\def\bLambda{{\boldsymbol \Lambda}}

\newif\ifarxiv 
\arxivtrue 

\title{Unrolling Lanczos for Ideal Low-pass Graph Filter Approximation}
\name{Parham Eftekhar$^\dag$, Gene Cheung$^\dag$\thanks{The work of G. Cheung was supported in part by the Natural Sciences and Engineering Research Council of Canada (NSERC) RGPIN-2025-06252 and Manulife.}, Mingxiao Liu$^\star$, H. Vicky Zhao$^\star$, Xuejun Han$^\#$, Eugene Wen$^\#$}
\address{ $^\dag$York University, Canada 
~~~~~~ $^\star$Tsinghua University, China
~~~~~~ $^\#$Manulife, Canada }
\begin{document}
\ninept
\maketitle
\begin{abstract}
Low-pass (LP) filtering is a fundamental operation in graph signal processing (GSP). 
Among finite-order nodal-domain methods, Lanczos-based filtering provides more accurate approximations of ideal LP filters than Chebyshev polynomial methods. 
We show that the approximation of Lanczos filtering can be further improved through algorithm unrolling and data-driven parameter learning. 
The key insight is that, because ideal LP filtering is a projection operation into the low-frequency eigen-subspace $\cS_K$, instead of approximating individual eigen-pairs of a graph Laplacian $\L$ as done in classical Lanczos, an unrolled Lanczos network can directly approximate $\cS_K$.
Specifically, we first establish a theorem identifying properties of the Lanczos tridiagonal matrix $\T_m$ that promote accurate approximation of the low-frequency eigen-subspace $\cS_K$. 
Guided by this theory, we relax the orthogonality constraint on Lanczos vectors, resulting in Ritz vectors that better span $\cS_K$.
To ensure numerical stability, we constrain $\T_m$ to be similar to a symmetric matrix, thereby guaranteeing real-valued eigenvalues. 
Experimental results on random graphs and learned graphs in two natural language processing (NLP) tasks show that our unrolled Lanczos network achieves superior ideal LP filter approximation compared to the classical Lanczos method.
\end{abstract}
\section{Introduction}
\label{sec:intro}

Spectral filtering is at the heart of the field of \textit{graph signal processing} (GSP)
\cite{ortega18ieee,cheung18} that analyzes discrete signals residing on finite graph structures.
Specifically, defining graph frequencies as eigen-modes $\U$ of a chosen graph Laplacian matrix $\L = \U \bLambda \U^\top \in \mathbb{R}^{N \times N}$ for eigenvalues $\bLambda = \mathrm{diag}(\{\lambda_i\})$, a \textit{low-pass} (LP) filter $f(\lambda)$ is one that retains low frequency components of a signal $\x$ while suppressing high ones.
While one can compute the filter output in the graph transform domain as $\y = \U f(\bLambda) \U^\top \x$, eigen-decomposition of $\L$ to obtain eigen-modes $\U$ has $\cO(N^3)$ worst-case complexity. 
To mitigate the cubic complexity of eigen-decomposition, spectral graph filtering is often performed directly in the nodal domain.

For ideal LP graph filter $f(\lambda)$ with cutoff frequency $\omega_K$, \ie,  
\begin{align}
f(\lambda) = \left\{
\begin{array}{ll}
1 & \mbox{if}~~ \lambda \leq \omega_K \\
0 & \mbox{o.w.}
\end{array} 
\right. ,
\label{eq:idealLP}
\end{align}
it was shown \cite{Susnjara2015,Vu2021} that Lanczos-based filtering \cite{golub12} yields more accurate finite-order approximations than Chebyshev polynomial filtering \cite{Rivlin1990}, due to the discontinuity of the frequency response $f(\lambda)$ at cutoff $\omega_K$.
However, the Krylov basis produced by finite-order Lanczos is constrained by the Galerkin orthogonality condition \cite{golub12}, designed to optimize approximation of individual eigen-pairs\footnote{Another popular method to approximate extreme eigen-modes of a large symmetric matrix is \textit{Locally Optimal Block Preconditioned Conjugate Gradient} (LOBPCG) \cite{knyazev2001toward}. Unrolling of the iterative LOBPCG and identifying parameters for data-driven learning is another possible research direction; we leave this for future work.}.

In this paper, leveraging the growing trend in \textit{algorithm unrolling} that converts iterative optimization algorithms to interpretable deep networks \cite{monga21,yu23nips}, we show how unrolling the Lanczos algorithm into a neural net to learn parameters in a data-driven manner can further improve approximation for ideal LP graph filtering.
Beyond classical GSP, ideal LP graph filtering is of particular interest in the deep learning era: it was shown in \cite{Do2024} that a combination of graph learning and LP graph filtering constitutes a \textit{self-attention} mechanism \cite{Vaswani2017}, and hence unrolling a graph filtering algorithm results in a \textit{transformer} \cite{Do2024,Cai2026,Yao2026,Qi2026}. 
Thus, improved LP graph filtering via Lanczos unrolling holds promise for more efficient transformer-style architectures for a broad range of applications.

Our key insight is the following: 
\textit{because ideal LP filtering is a projection operation into the low-frequency eigen-subspace $\cS_K = \mathrm{span}(\{\u_i\}_{i=1}^K)$ of $\L$, instead of approximating individual eigen-pairs $\{(\lambda_i,\u_i)\}_{i=1}^K$ of $\L$ as done in classical Lanczos, an unrolled Lanczos network can directly approximate $\cS_K$}.
Specifically, we first develop a theorem stating conditions of the Lanczos tridiagonal matrix $\T_m$ to well approximate $\cS_K$.
Guided by this theory, we relax the orthogonality constraint on the Lanczos vectors used to construct $\T_m$, so that the resulting Ritz vectors can better span $\cS_K$.
Moreover, we preserve spectral stability by constraining $\T_m$ to be similar to a symmetric matrix, thereby guaranteeing real-valued eigenvalues.
Experimental results on random graph models and learned similarity graphs interpreted as self-attention for two natural language processing (NLP) tasks  show that our proposed unrolled Lanczos network achieves superior approximation of ideal LP filtering compared to the conventional Lanczos  approximation.

\section{Preliminaries}

\subsection{GSP Definitions}

A graph $\cG(\cV,\cE)$ contains a node set $\cV = \{1, \ldots, N\}$ and an edge set $\cE$, where $(i,j) \in \cE$ means nodes $i,j \in \cV$ are connected by an edge with weight $w_{i,j} \in \mathbb{R}$.
Denote by $\W \in \mathbb{R}^{N \times N}$ an \textit{adjacency matrix}, where $W_{i,j} = w_{i,j}$ if $(i,j) \in \cE$, and $W_{i,j} = 0$ otherwise.
Denote by $\D \in \mathbb{R}^{N \times N}$ a diagonal \textit{degree matrix}, where $D_{i,i} = \sum_{j} W_{i,j}$.
Denote by $\L \in \mathbb{R}^{N \times N}$ the \textit{combinatorial graph Laplacian} $\L \triangleq \D - \W$.
If the graph is undirected, \ie, $w_{i,j} = w_{j,i}, \forall i,j$, then $\W$ and $\L$ are symmetric.
If self-loops exist, \ie, $w_{i,i} \neq 0, \exists i$, then the \textit{generalized graph Laplacian} $\cL \triangleq \D - \W + \mathrm{diag}(\W)$ is used.
We use $\L$ and $\cL$ for both positive graphs and signed graphs with both positive and negative edges.

$\L$ (or $\cL$) is provably \textit{positive semi-definite} (PSD) if all edge weights are positive, \ie, $w_{i,j} \geq 0, \forall i,j$ \cite{cheung18}.
For a signed graph with both positive and negative edges, $\L$ (or $\cL$) may be indefinite.
Given $\L$ (or $\cL$) is real and symmetric, it can be eigen-decomposed into $\L = \U \bLambda \U^\top$ by the Spectral Theorem (Chapter 2, \cite{Horn2012}), where $\U \in \mathbb{R}^{N \times N}$ is the eigen-matrix with orthonormal eigenvectors $\{\u_i\}_{i=1}^N$ of $\L$ as columns, and $\bLambda = \mathrm{diag}(\{\lambda_i\}) \in \mathbb{R}^{N \times N}$ is a diagonal matrix with ordered eigenvalues $\{\lambda_i\}_{i=1}^N$ along its diagonal. 
Assuming the graph $\cG$ is positive, it is common in GSP to interpret the $i$-th eigen-pair $(\lambda_i, \u_i)$ of PSD $\L$ as the $i$-th frequency and Fourier mode for $\cG$.
Given a \textit{graph signal} $\x \in \mathbb{R}^N$, $\balpha = \U^\top \x$ is the \textit{Graph Fourier Transform} (GFT), wher $\alpha_i$ is the $i$-th GFT coefficient of signal $\x$.

\section{Review \& Analysis of Lanczos}
\label{sec:formulate}

We first review the Lanczos algorithm (Chapter 10, \cite{golub12}) commonly used to approximate extreme eigen-pairs of a large symmetric matrix $\L \in \mathbb{R}^{N \times N}$. 
We then develop theory to characterize conditions under which the reduced-dimension tridiagonal matrix $\T_m \in \mathbb{R}^{m \times m}$, $m \ll N$, produces lifted Ritz vectors that well approximate the smallest eigen-pairs of $\L$.

\subsection{Review of the Lanczos Algorithm}

Given a unit-norm initial vector $\v_1$ and a real symmetric matrix $\L$, the $m$-th \textit{Krylov space} is defined as 
\begin{align}
\cK_m(\L,\v_1) \triangleq \mathrm{span}(\{\v_1, \L \v_1, \ldots, \L^{m-1} \v_1\}). 
\end{align}
The Lanczos algorithm iteratively computes $m$ orthonormal vectors $\V_m = [\v_1, \ldots, \v_m]$---called the \textit{Lanczos vectors}---that span the Krylov space.
Specifically, $\L$ and $\V_m$ are related by

\vspace{-0.1in}
\begin{footnotesize}
\begin{align}
\V_m^\top \L \V_m &= \T_m = \left[ \begin{array}{cccccc}
\alpha_1 & \beta_1 & & & & 0 \\
\beta_1 & \alpha_2 & \beta_2 & & & \\
 & \beta_2 & \alpha_3 & \ddots & & \\
& & \ddots & \ddots & \beta_{m-2} & \\
& & & \beta_{m-2} & \alpha_{m-1} & \beta_{m-1} \\
0 & & & & \beta_{m-1} & \alpha_m
\end{array}
\right]
\label{eq:Lanczos_matrix}
\end{align}
\end{footnotesize}\noindent
where $\T_m$ is an $m \times m$ symmetric tridiagonal matrix. 

Another interpretation of \eqref{eq:Lanczos_matrix} is that $\T_m$ is the \textit{Rayleigh-Ritz projection} of $\L$ onto the Krylov space $\cK_m(\L,\v_1)$. 
Denote by $\{(\theta_i,  \y_i)\}_{i=1}^K$ the eigen-pairs of $\T_m$ corresponding to its $K$ smallest eigenvalues. 
Define the associated \textit{Ritz vectors}: 
\begin{align}
\x_i = \V_m \y_i, \qquad i=1,\ldots,K.    
\end{align}
The resulting Ritz pairs $\{(\theta_i,\x_i)\}_{i=1}^K$ satisfy the known \textit{Galerkin orthogonality condition} \cite{golub12}: 
\begin{align} 
\L \x_i - \theta_i \x_i \perp \cK_m(\L,\v_1), \qquad i=1,\ldots,K.
\label{eq:Galerkin}
\end{align} 
In words, \eqref{eq:Galerkin} states that the error of the eigen-pair approximation $\L \x_i - \theta_i \x_i$ is orthogonal to the Krylov space $\cK_m(\L, \v_1)$.
Thus, the first $K$ Ritz pairs provide accurate approximations to the first $K$ eigen-pairs of $\L$ given the Krylov space of dimension $m$.

\subsection{Coefficient Analysis of $\T_m$}

$\T_m$ in \eqref{eq:Lanczos_matrix} is defined by $\alpha_i$'s on the diagonal and $\beta_i$'s on the super- and sub-diagonals. 
We study conditions on the coefficients $\alpha_i$ and $\beta_i$ that promote accurate coverage of the low-frequency eigen-subspace $\cS_K = \mathrm{span}(\{\u_i\}_{i=1}^K)$.

\begin{theorem}[Low-Frequency Subspace Coverage and Ritz-Value Bounds]
\label{thm:low_frequency_bounds}

Let $\L\in\R^{N\times N}$ be a real symmetric matrix with eigenvalues \[ \lambda_1\leq\cdots\leq\lambda_K < \lambda_{K+1}\leq\cdots\leq\lambda_N. \] 
Let $\V_m=[\v_1,\ldots,\v_m]$ be the Lanczos basis, and let $\T_m=\V_m^\top \L \V_m$ be the associated Lanczos tridiagonal matrix. 
Suppose that 
\begin{align}
\alpha_i \le \lambda_{K+1}-\varepsilon, \qquad i=1,\ldots,m,    
\end{align}
for some $\varepsilon>0$, and 
\begin{align}
|\beta_i|\ge \beta_{\min}>0, \qquad i=1,\ldots,m-1.    
\end{align}
Then: 1) Denote by $\cS_K = \mathrm{span}(\{u_i\}_{i=1}^K)$ the eigen-subspace associated with the $K$ smallest eigenvalues of $\L$, and denote by
\begin{align}
\cP_K = \sum_{j=1}^{K} \u_j \u_j^\top     
\end{align}
the corresponding orthogonal projector. 
Then 
\begin{align}
\mathrm{tr}(\cP_K \V_m \V_m^\top) \ge \frac{m\varepsilon}{\lambda_{K+1}-\lambda_1}.    
\end{align}
2) Denote by $ \theta_1\le\cdots\le\theta_K$ the $K$ smallest eigenvalues of $\T_m$. 
Then
\begin{align}
\theta_1 \le (\lambda_{K+1}-\varepsilon) - 2\beta_{\min} \cos\!\left(\frac{\pi}{m+1}\right).    
\end{align}
Therefore, smaller diagonal coefficients $\alpha_i$ and larger off-diagonal coefficients $|\beta_i|$ are desirable for recovering the low-frequency eigensubspace $\cS_K$: the former increases the energy of the Lanczos basis within $\cS_K$, while the latter encourages smaller Ritz values associated with the low-frequency spectrum.
\end{theorem}

\ifarxiv 
See Appendix\;\ref{append:proof1} for a proof.
\else
See Appendix\;A in \cite{Eftekhar2026} for a proof.
\fi

\section{Lanczos Algorithm Unrolling}

Instead of employing the Lanczos algorithm directly, we unroll its iterations into neural layers for data-driven parameter tuning. 
We first write the Lanczos algorithm in iterative algebraic form:
\begin{align}
    \w_j &= \L \v_j, ~~~ \alpha_j = \v_j^\top \w_j
    \label{eq:Lanczos_alpha}\\
    \r_j &= \w_j - \alpha_j \v_j - \beta_{j-1} \v_{j-1} 
    \label{eq:Lanczos_iterative} \\
    \beta_j &= \|\r_j\|, ~~~ \v_{j+1} = \r_j/\beta_j
    \label{eq:Lanczos_beta}
\end{align}
where \eqref{eq:Lanczos_iterative} can be interpreted as an orthogonalization step, so that $\r_j$ (and subsequent normalized $\v_{j+1}$) is orthogonal to previously computed $\v_j, \v_{j-1}, \ldots, \v_1$, akin to the \textit{Gram-Schmidt procedure} (Chapter 5, \cite{golub12}).

\subsection{Relaxing the Lanczos Recurrence}

Classical Lanczos seeks Ritz pairs satisfying the described Galerkin orthogonality condition \eqref{eq:Galerkin}, and therefore accurately approximate individual eigen-pairs of $\L$. 
However, ideal LP filtering depends only on the low-frequency eigen-subspace $\cS_K=\mathrm{span}(\{\u_i\}_{i=1}^K)$ through the projector $\cP_K = \U_K \U_K^\top$, and not on the individual eigenvalues or eigenvectors. 
Thus, \textit{exact recovery of the first $K$ eigen-pairs is unnecessary, provided that the subspace spanned by the Ritz vectors closely matches $\cS_K$}. 
This crucial observation motivates relaxing the classical Lanczos recurrence and directly learning a Krylov basis that better captures the target low-frequency eigen-subspace.

Specifically, instead of \eqref{eq:Lanczos_iterative}, we relax the orthogonality condition and introduce learnable parameters,  $\gamma^{(1)}_j$ and $\gamma^{(2)}_j$, at each step $j$:
\begin{align}
\r_j &= \w_j - \gamma^{(1)}_{j} \alpha_j \v_j - \gamma^{(2)}_{j} \beta_{j-1} \v_{j-1} .
\label{eq:Lanczos_iterative2}
\end{align}
In words, \eqref{eq:Lanczos_iterative2} enables $\gamma^{(1)}_j$ and $\gamma^{(2)}_j$ to \emph{steer} the evolution of the Krylov basis towards improved coverage of the eigen-space $\cS_K$. 
In contrast, classical Lanczos seeks solely to enforce orthogonality and optimize eigen-pair approximation. 

The resulting vectors are generally not mutually orthogonal. 
Consequently, the tridiagonal matrix $\T_m$ is no longer symmetric: 

\vspace{-0.05in}
\begin{footnotesize}
\begin{align}
\T_m = \left[ 
\begin{array}{ccccc}
\gamma^{(1)}_1 \alpha_1 & \gamma^{(2)}_2 \beta_1 & & &  0 \\
 \beta_1 & \gamma^{(1)}_2 \alpha_2 & \gamma^{(2)}_3 \beta_2 & & \\
 & \beta_2 & \ddots & \ddots & \\
& & \ddots & \gamma^{(1)}_{m-1} \alpha_{m-1} & \gamma^{(2)}_m \beta_{m-1} \\
0 & & & \beta_{m-1} & \gamma^{(1)}_m \alpha_m
\end{array}
\right] .
\label{eq:Tm_relaxed}
\end{align}
\end{footnotesize}\noindent
Further, the relation $\T_m = \V_m^\top\L\V_m$ no longer holds. 
Thus, $\T_m$ is not a Rayleigh--Ritz projection of $\L$ onto the Krylov subspace, and the resulting Ritz vectors do not satisfy the Galerkin orthogonality condition. 
Nevertheless, because each iterate is generated from $\L\v_j$, $\v_j$, and $\v_{j-1}$, the vectors $\{\v_1,\ldots,\v_m\}$ continue to span the Krylov subspace, provided that no breakdown occurs.


\subsection{Numerical Stability of $\T_m$}

One caveat is that a general nonsymmetric tridiagonal matrix may possess complex eigenvalues. 
To ensure that $\T_m$'s spectrum is real, we impose the constraint  $\gamma^{(2)}_j > 0, \forall j$. 
Under this condition, $\T_m$ is diagonally similar to a symmetric tridiagonal matrix and therefore has only real eigenvalues.
We formalize this in the following lemma.

\begin{lemma} 
Suppose 
\begin{align}
\T_m = \left[ \begin{array}{ccccc} a_1 & c_1 & & & 0 \\ b_2 & a_2 & c_2 & & \\ & b_3 & \ddots & \ddots & \\ & & \ddots & a_{m-1} & c_{m-1} \\ 0 & & & b_m & a_m \end{array} \right]
\end{align}
is a real tridiagonal matrix satisfying
\begin{align}
b_j c_{j-1} > 0, \qquad j=2,\ldots,m. 
\end{align} 
Then $\T_m$ is diagonal similar to a symmetric tridiagonal matrix. Consequently, all eigenvalues of $\T_m$ are real. 
\end{lemma} 

\ifarxiv 
See Appendix\;\ref{append:proof2} for a proof.
\else
See Appendix\;B in \cite{Eftekhar2026} for a proof.
\fi
Thus, given $ b_j=\gamma_j^{(2)}\beta_{j-1}$ and $c_{j-1}=\beta_{j-1}$, it follows that $b_j c_{j-1} = \gamma_j^{(2)}\beta_{j-1}^2$. 
Therefore, if $\gamma_j^{(2)}>0$ for all $j$, then all eigenvalues of $\T_m$ in \eqref{eq:Tm_relaxed} are real.

\subsection{Implementing Unrolled Lanczos Module}

Guided by Theorem~\ref{thm:low_frequency_bounds}, we design a lightweight module that adapts the Lanczos recurrence to better approximate the low-frequency eigensubspace $\cS_K$. 


Ideally, the initial vector $\v_1$ should contain mostly components of the first $K$ eigenvectors $\{\u_i\}_{i=1}^K$ of $\L$, \ie, $\v_1 \approx \sum_{i=1}^K q_i \u_i$, for $|q_i| > 0$. 
Doing so means $\L^{j-1} \v_1 \approx  \sum_{i=1}^K q_i \lambda_i^{j-1} \u_i$, and Lanczos vector $\v_j$'s can be informative.
Obviously, $\{\u_i\}_{i=1}^K$ of $\L$ are unknown \textit{a priori}, and we cannot construct $\v_1$ directly from $\{\u_i\}_{i=1}^K$.
Instead, assuming $\L$ is PSD, we compute a probe vector $\v_1$ using an iid random Gaussian vector $\z \sim \cN(0,\I)$:
\begin{align}
\v_1 = \frac{(\I + \sum_{t=1}^T a_t \L^t)^{-p} \z}{\|(\I + \sum_{t=1}^T a_t \L^t)^{-p} \z\|_2}, ~~ a_t \geq 0, ~\forall t
\end{align}
where $p \in \mathbb{Z}_+$, and $\{a_t\}_{t=1}^T$ are learnable filter coefficients constrained to be non-negative. 
$(\I + \sum_t a_t \L^t)^{-p}$ is LP filter.

\begin{lemma}
Let $\L \succeq 0$, $a_t \geq 0$ for $t=1,\ldots,T$, and $p\in\mathbb{N}$.
Then 
\begin{align}
\H(\L) = \left( \I + \sum_{t=1}^{T} a_t \L^t\right)^{-p}    
\end{align}
is a low-pass graph filter with respect to $\L$.    
\end{lemma}

\ifarxiv 
See Appendix\;\ref{append:proof3} for a proof.
\else
See Appendix\;C in \cite{Eftekhar2026} for a proof.
\fi
To enforce $a_t \geq 0$, we employ the softplus function $a_t = \log ( 1 + e^{b_t})$ and learn unconstrained $b_t$.



To learn appropriate parameters $\gamma^{(1)}_i$ and $\gamma^{(2)}_i$ in \eqref{eq:Lanczos_iterative}, guided by Theorem\;\ref{thm:low_frequency_bounds} that states smaller $\alpha_i$ and larger $|\beta_i|$ in $\T_m$ promote recovery of $\cS_K$, we constrain $\gamma_i^{(1)}\leq 1$ and $\gamma_i^{(2)}\geq 1$, so that $\gamma_i^{(1)}\alpha_i$ is reduced and $\gamma_i^{(2)}|\beta_{i-1}|$ is amplified relative to classical Lanczos. We enforce these constraints by parameterizing the relaxation coefficients at layer $i$ as
\begin{align}
\gamma_i^{(1)} &= 1 - \bigl(\upsilon_i^{(1)}\bigr)^2, ~~~~~~
\gamma_i^{(2)} = 1 + \bigl(\upsilon_i^{(2)}\bigr)^2,
\end{align}
where $\upsilon_i^{(1)}$ and $\upsilon_i^{(2)}$ are unconstrained learnable parameters. 

\section{Experiments}
\label{sec:results}

\subsection{Simulation Experiments}

We first evaluate the proposed unrolled Lanczos model in settings where each graph is fixed and known. For these experiments, we use a synthetic stochastic block model (SBM) dataset and the PROTEINS dataset, as described below.

\noindent\textbf{Synthetic Dataset.} We synthesize a family of 50 stochastic block model (SBM) graphs, each comprising 100 nodes, with within-community and between-community edge probabilities set to $p_{in}=0.3$ and $p_{out}=0.02$, respectively. The generated graphs are partitioned into training, validation, and test sets using a 60/20/20 split, resulting in 30, 10, and 10 graphs, respectively.

\noindent\textbf{PROTEINS Dataset.} The PROTEINS dataset, obtained from the TUDataset collection~\cite{morris2020tudataset}, contains 1,113 graphs with varying numbers of nodes. To construct a fixed-size dataset, we discard graphs with fewer than 50 nodes and use breadth-first search (BFS) to extract a connected 50-node subgraph from each remaining graph. This preprocessing yields 238 graphs, which we partition into training, validation, and test sets of 142, 48, and 48 graphs, respectively.


Motivated by Theorem~\ref{thm:low_frequency_bounds}, we adopt a training objective that encourages small diagonal coefficients $\alpha_i$ and large off-diagonal magnitudes $|\beta_i|$, aiming to improve the network's approximation of ideal low-pass graph filtering. Since we assume $\L$ is PSD, $\alpha_i=\v_i^\top\L\v_i\geq0$; hence, penalizing $\alpha_i^2$ encourages small $\alpha_i$.
\begin{align}
\mathcal{J}_{\mathrm{theory}}
&= \sum_{j=1}^{m} |\alpha_j|^2
- \lambda \sum_{j=2}^{m} |\beta_{j-1}|^2.
\label{eq:theory_loss}
\end{align}

We evaluate approximation accuracy using the subspace error. 
Denote by $\U_K=[\u_1,\ldots,\u_K]$ the eigen-submatrix containing the first $K$ orthonormal eigenvectors of $\L$ as columns.
Denote by $\widehat{\U}_K$ a $K$-dimensional basis obtained from the lifted vectors $[\V_m\y_1,\ldots,\V_m\y_K]$. 
The subspace error and its normalized squared form are
\begin{align}
\mathcal{E}_{\mathrm{sub}}
&= \| (\I-\widehat{\U}_K\widehat{\U}_K^\top)\U_K \|_F,
\\
\overline{\mathcal{E}}_{\mathrm{sub}}^{\,2}
&= \frac{1}{K}\mathcal{E}_{\mathrm{sub}}^2.
\label{eq:subspace_error}
\end{align}
This metric measures the portion of the target eigen-subspace not captured by the approximation, independently of eigenvector signs or the choice of basis.
The normalized squared error lies in $[0,1]$, with smaller values indicating greater accuracy and zero indicating identical subspaces.

Although the subspace error is well suited for evaluation, using it directly as the training objective would require eigen-decomposition of $\T_m$ at every forward pass, increasing computational overhead and potentially introducing numerical instability when back-propagating through its eigenvectors. 
The theory-guided objective in \eqref{eq:theory_loss} avoids this decomposition and therefore provides a more efficient and numerically stable training criterion.


We implement the proposed model in PyTorch and train it for 50 epochs with a batch size of 4 on an NVIDIA RTX PRO 5000 Blackwell GPU. For each dataset, we select hyperparameters using the validation set through a grid search over inverse-polynomial powers \(\{3,4,5,6\}\), numbers of conjugate-gradient (CG) steps \(\{10,20,30\}\) to solve linear systems for $\v_1$, polynomial degrees \(\{7,9,11\}\), and learning rates \(\{10^{-3},3\times10^{-3}\}\). 

Table~\ref{tab:simulation} reports the subspace error for approximating the low-frequency eigen-subspace spanned by the first \(K\) eigenvectors, for different \(K\). 
We set the number of unrolled Lanczos layers to \(m=2K\). 
We compare the trained network with its initialization, where \(\gamma_i^{(1)}=\gamma_i^{(2)}=1\) and the recurrence reduces to standard Lanczos. 
We see that unrolled Lanczos noticeably reduces the error over standard Lanczos for both datasets and all values of \(K\).

\subsection{NLP Experiments}
In the second set of experiments, we evaluate the proposed model on two real-world natural language processing (NLP) tasks, QNLI and MRPC, from the GLUE benchmark~\cite{wang2018glue}. QNLI is a question–answer natural language inference task, where the model determines whether a given context contains the answer to a question, while MRPC is a paraphrase identification task that determines whether two sentences are semantically equivalent.


We consider a conventional Transformer block comprising a self-attention module followed by a residual connection and layer normalization, and a feed-forward network (FFN) followed by another residual connection and layer normalization. 
Inspired by \cite{Do2024}, we replace the self-attention module with our pair of modules: a graph-learning module and an ideal low-pass filter approximation module. Given a sequence of token embeddings $\{\x_i\}_{i=1}^{N}$, with $\x_i\in\mathbb{R}^{E}$, a feature projection $F:\mathbb{R}^{E}\rightarrow\mathbb{R}^{D}$ produces lower-dimensional representations $\f_i=F(\x_i)\in\mathbb{R}^{D}$, where $D\ll E$. For each pair of tokens $i$ and $j$, we measure their feature similarity using the squared \textit{Mahalanobis distance}
\begin{align}
d_{i,j} = (\f_i - \f_j)^\top \M (\f_i - \f_j) \geq 0,
\end{align}
where the symmetric metric matrix $\M\succeq0$ is PSD. 
Then, feature distance $d_{i,j}$ is mapped to non-negative edge weight $w_{i,j} \in [0,1]$ via an exponential kernel:
\begin{align}
w_{i,j} = \exp (-d_{i,j}) .
\label{eq:weightExp}
\end{align}
Using these learned edge weights, we construct a combinatorial graph Laplacian $\L$, which serves as input to the ideal low-pass filter approximation module. The Laplacian is processed by the unrolled Lanczos network, and the resulting tridiagonal matrix $\T_m$ is eigen-decomposed to approximate the ideal low-pass filter.

We train the model for 10 epochs using a learning rate of $10^{-4}$. The number of unrolled Lanczos layers is selected from $\{15,20,25\}$ based on validation performance. Table~\ref{tab:nlp} compares our model with a native self-attention baseline and a standard Lanczos baseline, obtained by fixing all tunable parameters $\gamma$ to one. 
We report both accuracy and F1 score.
We see that our proposed unrolled Lanczos outperforms standard Lanczos noticeably for both NLP tasks, and outperforms native self-attention as well for most metrics.


\begin{table}[t]
\centering
\caption{Comparison on synthetic (s) and PROTEINS (P) datasets for the first $K$ eigenvectors associated with the smallest eigenvalues.}
\label{tab:simulation}
\vspace{-0.1in}
\begin{small}
\begin{tabular}{lccc}
\toprule
Method & $K=6$ & $K=8$ & $K=10$  \\
\midrule
Unrolled Lanczos (s) & 0.439004 & 0.423422 & 0.367661 \\
Standard Lanczos (s) & 0.542560 & 0.503438 & 0.480293 \\
Unrolled Lanczos (P) & 0.358199 & 0.472659 & 0.433804 \\
Standard Lanczos (P) & 0.638954 & 0.553640 & 0.481555 \\
\bottomrule
\end{tabular}
\end{small}
\end{table}

\begin{table}[t]
\centering
\caption{Comparison on QNLI and MRPC in Accuracy / F1.}
\label{tab:nlp}
\vspace{-0.1in}
\begin{small}
\begin{tabular}{lcc}
\toprule
Method & QNLI & MRPC \\
\midrule
\textbf{Unrolled Lanczos} & 61.82 / \textbf{65.75} & \textbf{69.61 / 81.76} \\
Standard Lanczos & 55.43 / 30.84 & 68.87 / 81.41 \\
Self-Attention & \textbf{61.85} / 62.15 & 68.38 / 81.22 \\
\bottomrule
\end{tabular}
\end{small}
\end{table}

\section{Conclusion}
\label{sec:conclude}

To efficiently implement ideal low-pass (LP) graph filtering without computationally expensive eigen-decomposition, we unroll the iterative Lanczos algorithm into neural layers to enable data-driven parameter learning. 
The key observation is that ideal LP filtering is a projection onto the low-frequency eigen-subspace $\cS_K=\mathrm{span}(\{\u_i\}_{i=1}^{K})$. 
Motivated by this observation, we relax the orthogonality constraints imposed by classical Lanczos, which are designed for approximation of individual eigen-pairs, and instead learn Ritz vectors that more directly capture the target eigen-subspace $\cS_K$. 
Experimental results show that our unrolled Lanczos network achieves more accurate ideal LP graph filter approximation than the classical Lanczos method.

\ifarxiv 

\appendix

\section{Proof of Theorem\;1}
\label{append:proof1}

\begin{proof} 
For each Lanczos vector $\v_i$, decompose 
\begin{align}
\v_i=\cP_K\v_i+\z_i, \qquad \z_i=(\I-\cP_K)\v_i.     
\end{align}
Since the eigenvalues of $\L$ restricted to $\cS_K$ are at least $\lambda_1$ and those restricted to $\cS_K^\perp$ are at least $\lambda_{K+1}$, we have 
\begin{align}
\alpha_i = \v_i^\top \L \v_i \ge \lambda_1\|\cP_K\v_i\|_2^2 + \lambda_{K+1}\|\z_i\|_2^2.    
\end{align}
Using $\|\cP_K\v_i\|_2^2+\|\z_i\|_2^2=1$ and the assumption $\alpha_i\le \lambda_{K+1}-\varepsilon$, we obtain 
\begin{align}
\|\cP_K\v_i\|_2^2 \ge \frac{\varepsilon}{\lambda_{K+1}-\lambda_1}.    
\end{align}
Summing over $i=1,\ldots,m$ yields 
\begin{align}
\mathrm{tr}(\cP_K\V_m\V_m^\top) = \sum_{i=1}^{m}\|\cP_K\v_i\|_2^2 \ge \frac{m\varepsilon}{\lambda_{K+1}-\lambda_1}.    
\end{align}
Next, write 
\begin{align}
\T_m=\mathrm{diag}(\alpha_1,\ldots,\alpha_m)+\T_o,     
\end{align}
where $\T_o$ is obtained from $\T_m$ by setting all diagonal entries to zero. 
Since $\alpha_i\le \lambda_{K+1}-\varepsilon$, 
\begin{align}
\mathrm{diag}(\alpha_1,\ldots,\alpha_m) \preceq (\lambda_{K+1}-\varepsilon)\I,    
\end{align}
and hence 
\begin{align}
\T_m \preceq (\lambda_{K+1}-\varepsilon)\I+\T_o.     
\end{align}
By eigenvalue monotonicity, 
\begin{align}
\theta_1 \le (\lambda_{K+1}-\varepsilon)+\lambda_{\min}(\T_o).    
\end{align}
Let $\T_{\min}$ denote the tridiagonal matrix with zero diagonal entries and all off-diagonal entries equal to $\beta_{\min}$. Since $|\beta_i|\ge \beta_{\min}$, a Rayleigh quotient comparison yields 
\begin{align}
\lambda_{\min}(\T_o) \le \lambda_{\min}(\T_{\min}).    
\end{align}
The eigenvalues of $\T_{\min}$ are 
\begin{align*}
2\beta_{\min}\cos\!\left(\frac{k\pi}{m+1}\right), \qquad k=1,\ldots,m,    
\end{align*}
and therefore 
\begin{align}
\lambda_{\min}(\T_{\min}) = -2\beta_{\min} \cos\!\left(\frac{\pi}{m+1}\right).    
\end{align}
Combining the above inequalities gives 
\begin{align}
\theta_1 \le (\lambda_{K+1}-\varepsilon) - 2\beta_{\min} \cos\!\left(\frac{\pi}{m+1}\right).    
\end{align}
\end{proof}

\section{Proof of Lemma\;1}
\label{append:proof2}

\begin{proof} 
Define the diagonal matrix $\D = \mathrm{diag}(d_1,\ldots,d_m),$ with 
\begin{align}
d_1 = 1, \qquad d_j = d_{j-1} \sqrt{\frac{b_j}{c_{j-1}}}, \quad j=2,\ldots,m.   
\end{align}
Since $b_j c_{j-1}>0$, the ratio $b_j/c_{j-1}$ is positive and hence every $d_j$ is real and nonzero. 
Consider the similarity transform
\begin{align}
\S = \D^{-1}\T_m\D.    
\end{align}
The diagonal entries of $\S$ equal those of $\T_m$, while for $j=2,\ldots,m$, 
\begin{align}
S_{j,j-1} = \frac{b_j d_{j-1}}{d_j} = \sqrt{b_j c_{j-1}},
\end{align}
and 
\begin{align}
S_{j-1,j} = \frac{c_{j-1} d_j}{d_{j-1}} = \sqrt{b_j c_{j-1}}.
\end{align}
Hence $\S$ is symmetric: 
\begin{small}
\begin{align}
\S = \left[ \begin{array}{ccccc} a_1 & \sqrt{b_2c_1} & & & 0 \\ \sqrt{b_2c_1} & a_2 & \sqrt{b_3c_2} & & \\ & \sqrt{b_3c_2} & \ddots & \ddots & \\ & & \ddots & a_{m-1} & \sqrt{b_mc_{m-1}} \\ 0 & & & \sqrt{b_mc_{m-1}} & a_m \end{array} \right].
\end{align}
\end{small}\noindent
Since $\T_m$ is similar to the real symmetric matrix $\S$, $\T_m$ and $\S$ have the same eigenvalues. 
Because every real symmetric matrix has real eigenvalues, all eigenvalues of $\T_m$ are real. 
\end{proof}

\section{Proof of Lemma\;2}
\label{append:proof3}

\begin{proof}
Since $\L$ is real symmetric and PSD, it admits the eigen-decomposition $\L = \U \bLambda \U^\top$, where $\bLambda = \mathrm{diag}(\lambda_1,\ldots,\lambda_N)$ and $0\leq \lambda_1\leq \cdots \leq \lambda_N$. Using functional calculus, $H(\L) = \U \, \mathrm{diag}\!\bigl(h(\lambda_1),\ldots,h(\lambda_N)\bigr) \U^\top$, where the spectral response is
\begin{align}
h(\lambda) = \left(1+\sum_{t=1}^{T} a_t\lambda^t\right)^{-p}.
\end{align}
Because $a_t\geq0$ and $\lambda\geq0$, 
\begin{align}
h'(\lambda)
=-p\left(1+\sum_{t=1}^{T}a_t\lambda^t\right)^{-p-1}
\sum_{t=1}^{T}t\,a_t\lambda^{t-1}\leq0.
\end{align}
Hence, $h(\lambda)$ is non-increasing with respect to $\lambda$. 
Therefore, for any $\lambda_i \leq \lambda_j$, 
\begin{align}
h(\lambda_i)\geq h(\lambda_j),    
\end{align}
meaning that lower graph frequencies are amplified (or attenuated less) than higher graph frequencies. 
\end{proof}

\section*{\centering\normalsize DISCLAIMER}
This work is experimental research and does not represent the practices, positions, or views of any specific organization.

\fi

\begin{small}
\bibliographystyle{IEEEbib}
\bibliography{refs}
\end{small}

\end{document}